\documentclass{article}

\usepackage{arxiv}

\usepackage[utf8]{inputenc}
\usepackage[T1]{fontenc}
\usepackage{hyperref}
\usepackage{url}
\usepackage{booktabs}
\usepackage{amsfonts}
\usepackage{nicefrac}
\usepackage{microtype}
\usepackage{graphicx}
\usepackage{amsmath,amssymb}
\usepackage{bm}
\usepackage{amsthm}
\usepackage{multirow}
\usepackage{float}

\newtheorem{theorem}{Theorem}
\newtheorem{assumption}{Assumption}

\newcommand{\Mtwo}{\mathcal{M}_2}
\newcommand{\DM}{\mathcal{D}_M}
\newcommand{\lcoh}{\lambda_{\mathrm{coh}}}
\newcommand{\lvis}{\lambda_{\mathrm{vis}}}
\newcommand{\Lop}{\mathcal{L}}
\newcommand{\Pioff}{\Pi_{\mathrm{off}}}
\newcommand{\LC}{\mathrm{LC}}
\newcommand{\Tr}{\operatorname{Tr}}
\newcommand{\Dcohlay}{\Delta_{\mathrm{coh}}^{\mathrm{lay}}}
\newcommand{\weff}{\widehat{\omega}}
\newcommand{\aeff}{\widehat{a}}
\newcommand{\eqdef}{\mathrel{\vcenter{\baselineskip0.5ex \lineskiplimit0pt
  \hbox{\scriptsize.}\hbox{\scriptsize.}}}=}
\providecommand{\ket}[1]{\lvert #1\rangle}
\providecommand{\bra}[1]{\langle #1\rvert}

\title{A Coherence Law for Trainability in Noisy Equivariant Quantum Neural Networks}

\author{
  Hassan Ugail \\
  Centre for Visual Computing and Intelligent Systems \\
  University of Bradford \\
  United Kingdom \\
  \And
  Newton Howard \\
  School of Individualized Study \\
  Rochester Institute of Technology \\
  United States \\
}

\begin{document}
\maketitle

\begin{abstract}
Symmetry provides a quantum neural network structure, but on its own it does not keep the network trainable once noise is present. We ask which physical quantity decides whether the gradients of an equivariant circuit survive decoherence, and we answer with a compact training law. Working with $U(1)$-equivariant brickwork circuits that conserve a charge, we find that two distinct effects govern a trainable gradient. Causality fixes where the gradient can live, confining it to the backward light cone of the readout inside the active charge sector. Coherence then determines how fast it decays through the contraction of the off-diagonal sector modes that the projected readout can actually observe. We prove a light-cone reduction that pins the noiseless gradient to the sector-restricted cone with a lower bound independent of the total qubit number, and we define a readout-visible aligned coherence rate as a Rayleigh quotient of the noise generator along the gradient-carrying mode. A perturbative open-system analysis turns this rate into a leading-order training law. Density-matrix simulations then confirm that the finite-noise degradation follows a single accumulated variable built from noise depth and coherence contraction, with a coefficient of determination of $0.979$. The sharpest test comes from a correlated-dephasing channel that has a large worst-case rate but a near-zero aligned rate. The law predicts no gradient loss for this channel, and none is seen. Sector coherence outperforms every standard channel diagnostic we compare it against, and the analysis identifies readout-visible sector coherence as the quantity that links equivariant architecture, open-system dynamics and noisy trainability.
\end{abstract}


\section{Introduction}\label{sec:intro}

Variational quantum learning rests on a simple operational requirement. A parametrised circuit can be trained only if the gradients of its output stay measurable at a realistic sampling cost \cite{biamonte2017,preskill2018,cerezo2021review,mitarai2018}. Quantum neural networks and the wider family of variational quantum algorithms encode information into a quantum state, process it through trainable gates, and read out an expectation value whose parameter derivatives drive optimisation \cite{havlicek2019,schuld2021,abbas2021}. When those derivatives fall below the sampling floor, the model cannot be optimised at all, and so the survival of the gradient signal is the question that governs near-term quantum learning.

Two distinct mechanisms can destroy this signal. One is the barren-plateau phenomenon, in which gradient second moments concentrate exponentially around zero as circuits grow wider, deeper, or overly expressive \cite{mcclean2018,cerezo2021bp,holmes2022,ragone2024}. The other is noise. Markovian decoherence acting throughout the circuit suppresses the gradient, drives noise-induced plateaus of its own, and caps the performance reachable by a noisy optimisation loop \cite{wang2021,stilckfranca2021,sharma2022}. The two are physically separate. One concerns the geometry of the noiseless landscape; the other, the open-system dynamics layered on top of it; and a model can be safe from the first while remaining acutely vulnerable to the second. What survives the noise is therefore the operational trainability question, and the natural way in is to ask which physical quantity governs that survival.

Symmetry has become one of the most productive ways to structure quantum neural networks against the first mechanism. Following the lead of classical geometric deep learning \cite{cohen2016,bronstein2021}, equivariant quantum architectures restrict the trainable gates to operations that commute with a symmetry group of the task \cite{larocca2022groupinv,nguyen2024,schatzki2024,meyer2023,ugail2026}. Under a $U(1)$ symmetry generated by the total charge, the Hilbert space splits into charge sectors, and an equivariant ansatz keeps a sector-supported input inside its sector throughout the noiseless evolution. Such circuits have reduced dynamical Lie algebras and can provably escape the exponential gradient concentration that afflicts generic deep circuits \cite{ragone2024,schatzki2024,pesah2021}, while sharpening generalisation from few samples \cite{caro2022,jerbi2023}. A broader principle is at work here. Preserving the structure a model is built on is often what protects its useful signal, a pattern that recurs across machine learning. It appears in deep face recognition under degraded or imperfect data \cite{ugail2026forensic,elmahmudi2019deep}, in transfer-learned visual attribution and interpretable prediction \cite{ugail2023raphael,ibrahim2025facial}, and in the dynamical modelling of organised complexity in physical systems \cite{ugail2025consciousness}. For the present circuits, the charge sectors supply both the physical structure and the natural arena for analysing trainability.

Symmetry preservation, however, is not noise protection. A channel can commute with the $U(1)$ action and still erase the quantum information that carries the gradient. The reason is that the gradient of a sector-projected readout lives not in sector populations but in coherences, the off-diagonal matrix elements within the active charge sector. Single-qubit dephasing makes the point concrete. It is phase-covariant; it moves no population between sectors, and a benchmark reporting only covariance defects and sector leakage would deem it harmless. Yet it contracts the intra-sector off-diagonal elements at a rate set by its strength, eroding the very matrix elements from which the gradient is built \cite{marvian2014ncomm,piani2016,baumgratz2014}. Three questions about a noisy equivariant circuit therefore need to be kept apart. One asks whether the channel respects the symmetry, a question of covariance. Another asks whether the population stays in the active sector, a question of sector retention. The last asks whether the off-diagonal signal that supports the gradient survives, a question of sector coherence. For common hardware channels the three answers diverge \cite{ugail2026access}, and it is the third that decides trainability.

The central claim of the paper is as follows. Gradients of a variational circuit are local response functions, and what governs them divides cleanly in two. Their support is fixed by causality, because in a local brickwork circuit a parameter can reach the readout only from inside the backward light cone of the measured observable, restricted to the relevant charge sector. Their amplitude is carried by intra-sector off-diagonal coherence, because the commutator structure behind a parameter derivative lives on the off-diagonal block of the sector. The trainability resource for a noisy equivariant quantum neural network is therefore neither symmetry covariance nor sector population. It is readout-visible sector coherence, the coherence on the off-diagonal modes of the active sector that the projected readout can actually see, within the light cone the readout can actually reach.

The training law is built around three entities. The first is a light-cone reduction for active gradients in $U(1)$-equivariant brickwork circuits. It shows that the active-gradient second moment is supported only inside the readout backward light cone, and that, under a non-degeneracy condition and a light-cone sector-weight condition on the fixed input, it carries a positive lower bound independent of the total qubit number at fixed depth and sector. The second is the readout-visible aligned coherence rate, a Rayleigh quotient of the noise generator taken along the gradient-carrying mode. We prove that this rate is bounded above by the worst-case scalar rate, with equality for restricted-isotropic noise, and from it we derive a perturbative coherence-loss law for phase-covariant Markovian noise. The third is empirical. Density-matrix simulations confirm that the finite-noise degradation obeys $\DM\propto(\gamma L)^{1.002}\lcoh^{0.964}$ with $R^2\simeq0.979$, that sector coherence outperforms standard channel diagnostics, and that a correlated-dephasing control with a large worst-case rate but a near-zero aligned rate produces no measurable gradient loss. That last control is decisive evidence that the aligned rate, rather than the worst-case rate, is the operative quantity.

It is worth fixing the scope before going further. The clean scalar law holds for the single-excitation sector under restricted-isotropic noise. Higher sectors and structured noise explicitly call for an aligned diagnostic, and coherent unitary miscalibration sits outside the dissipative family altogether. Within these limits, readout-visible aligned sector coherence is the quantity that bridges equivariant architecture, open-system dynamics and noisy trainability.

\section{Methods}\label{sec:methods}

\subsection{Model, readout and gradient observables}\label{sec:model}

We study a $U(1)$-equivariant quantum neural network on a cycle graph $C_n$ with periodic boundary conditions. The conserved charge is $\hat N=\sum_j(I-Z_j)/2$, and the circuit acts on a fixed input that lies in a single charge sector $r$, with $P_r$ the orthogonal projector onto that sector. The input is the localised computational-basis state $\rho_0=\ket{\psi_0}\bra{\psi_0}$ with $\ket{\psi_0}=\ket{1^r 0^{n-r}}$. The variational circuit is a depth-$L$ brickwork ansatz built from two gate families that commute with $\hat N$ \cite{nguyen2024,schatzki2024,larocca2022groupinv}. Each layer applies trainable single-qubit rotations $R_z(\theta_{\ell,j})$ on every site, followed by nearest-neighbour $XY$ hopping gates $\exp[-i\beta_{\ell,e}(X_{e_1}X_{e_2}+Y_{e_1}Y_{e_2})]$ on alternating edges of the cycle. The hopping parameters are fixed and treated as part of the architecture, so the only trainable degrees of freedom are the single-qubit rotation angles.

This study examines how noise degrades the parameter gradients of a fixed-state equivariant response rather than learning from a distribution of encoded data. The state the circuit processes is held fixed, and the object of interest is the second moment of the readout gradient with respect to the trainable angles. This is the quantity whose survival determines whether the circuit can be optimised at all, and it is well defined independently of any data-encoding map. The intra-sector off-diagonal coherence that carries the gradient is generated dynamically by the $XY$ hopping gates, which mix the in-sector basis states as the excitation propagates through the variational layers. The readout-visible coherence studied throughout the paper is therefore a property of the variational evolution acting on the fixed input, and the sector-supported state remains in the same charge sector throughout the noiseless circuit.

The output is the sector-projected local readout,
\begin{equation}
    O_B = P_r Z_0 P_r ,
    \label{eq:readout}
\end{equation}
and the noisy response is,
\begin{equation}
    f^\gamma_\theta=\Tr\!\bigl[O_B\,\Phi^\gamma_\theta(\rho_0)\bigr],
    \label{eq:prediction}
\end{equation}
where $\Phi^\gamma_\theta$ denotes the layered circuit channel with single-qubit Markovian noise of rate $\gamma$ applied once per layer \cite{lindblad1976,gks1976,breuerpetruccione2002}. Trainability is quantified by the prediction-gradient second moment,
\begin{equation}
    \Mtwo(\theta_i;\gamma)
    =
    \mathbb{E}_{\theta}
    \Bigl[\bigl(\partial_{\theta_i}f^\gamma_\theta\bigr)^2\Bigr],
    \label{eq:m2pred}
\end{equation}
with expectations taken over a small-box parameter prior of width $\delta_{\mathrm{init}}$ \cite{grant2019} and derivatives evaluated by the parameter-shift rule \cite{mitarai2018,schuld2019,wierichs2022}. The central response variable is the relative degradation,
\begin{equation}
    \DM(\gamma)
    =
    \log\!\bigl[\Mtwo(\theta_i;0)\,/\,\Mtwo(\theta_i;\gamma)\bigr],
    \label{eq:dm}
\end{equation}
which measures on a logarithmic scale how much active gradient strength the noise has removed.

\subsection{Active gradients and the backward light cone}\label{sec:lightcone}

Before asking how noise damages a gradient, it helps to ask where a gradient can live at all. The parameter-shift identity expresses a derivative as a commutator response,
\begin{equation}
    \partial_{\theta_i}f^0_\theta
    =
    -\tfrac{i}{2}\,\Tr\!\bigl[O_B(\ell_i)\,[Z_{q_i},\rho^\theta_{\ell_i}]\bigr],
    \label{eq:psr}
\end{equation}
where $O_B(\ell_i)$ is the readout evolved backwards in the Heisenberg picture. Each brickwork layer propagates operator support by at most one bond, so $O_B(\ell_i)$ is supported on the backward light cone $\LC(O_B)$ of the readout qubit, and outside that cone the commutator vanishes by causal locality. Because every gate commutes with $\hat N$, the trace restricts further to the charge sectors within the cone. A theorem makes this precise, subject to two conditions on the fixed background.

\begin{assumption}[Light-cone algebra non-degeneracy]\label{ass:nondegen}
For the fixed hopping background $\beta^\dagger$, the commutator between the light-cone-evolved readout operator and the active generator is not identically zero on the sector-restricted light-cone subspace.
\end{assumption}

\begin{assumption}[Light-cone sector weight]\label{ass:weight}
The fixed input places a non-vanishing charge weight inside the readout light cone, uniformly in system size. There is a constant $w_0>0$, independent of $n$ at fixed $L$ and $r$, such that the light-cone reduced state $\rho_{\LC}$ satisfies $\Tr[P^{\LC}_{q\geq1}\,\rho_{\LC}]\geq w_0$, where for the single-excitation sector $P^{\LC}_{q\geq1}=P^{\LC}_{q=1}$ projects onto the one-excitation subspace of the cone. The condition excludes inputs in which the excitation delocalises so strongly that its weight inside a fixed-size cone vanishes as $1/n$.
\end{assumption}

\begin{theorem}[Light-cone reduction of active gradients]\label{thm:locality}
Let $\LC(O_B)$ be the backward light cone of the readout, let $|\LC|$ be the number of qubits in it, and let $d_{\LC}=\sum_{q=0}^{r}\binom{|\LC|}{q}$ be the sector-restricted light-cone dimension. Active gradients are supported only inside the cone, since $\partial_{\theta_i}f^0_\theta=0$ whenever $\theta_i\notin\LC(O_B)$. Under Assumptions~\ref{ass:nondegen} and~\ref{ass:weight}, for any active parameter $\theta_i\in\LC(O_B)$ and the small-box prior of width $\delta_{\mathrm{init}}$,
\begin{equation}
    \Mtwo(\theta_i;0)
    \geq
    c_{\LC}\bigl(L,r,\delta_{\mathrm{init}},\beta^\dagger,O_B,w_0\bigr)
    >0 ,
    \label{eq:lcbound}
\end{equation}
where $c_{\LC}$ is independent of the total system size $n$ at fixed $L$ and $r$.
\end{theorem}

The locality statement that gradients vanish outside the cone needs no assumption on the input. The $n$-independent lower bound is the stronger claim and rests on Assumption~\ref{ass:weight}, which guarantees that the excitation keeps finite weight inside the fixed-size cone as $n$ grows. The theorem frames everything that follows. Only the light-cone parameters carry signal, so the question becomes what noise does to the gradient amplitude inside the cone, and that amplitude is carried by intra-sector coherence.

\begin{proof}[Proof of Theorem~\ref{thm:locality}]
We adopt the generator convention $R_z(\theta)=\exp(-i\theta Z/2)$, so the half-angle generator is $Z/2$ and the parameter-shift rule uses shifts of $\pi/2$. If parameter $\theta_i$ acts on qubit $q_i$ in layer $\ell_i$, the parameter-shift identity gives Equation~\eqref{eq:psr}. Parameters for which the commutator $[Z_{q_i},O_B(\ell_i)]$ vanishes identically on the sector-restricted Hilbert space are classified as inactive and excluded from Theorem~\ref{thm:locality}.

The proof begins with the sector-restricted light-cone reduction. Each brickwork layer propagates operator support by at most one bond, so the Heisenberg-evolved readout $O_B(\ell_i)$ is supported on the backward light cone $\LC(O_B)$. Outside the cone the commutator vanishes by causal locality, and because all gates commute with $\hat N$ the trace further restricts to the charge sectors $q\leq r$ within the cone, whose dimension is $d_{\LC}=\sum_{q=0}^{r}\binom{|\LC|}{q}$, replacing the full sector dimension $\binom{n}{r}$.

Next, the prior factorises over inactive coordinates. Writing $\theta=(\theta^{\LC},\theta^{\bar\LC})$ for parameters inside and outside the cone, the derivative for an active parameter depends only on $\theta^{\LC}$. The small-box prior is a product measure, so the coordinates outside the cone integrate to unit mass and
\begin{equation}
    \mathbb{E}_\theta\!\bigl[|\partial_{\theta_i}f^0_\theta|^2\bigr]
    =
    \int_{[-\delta_{\mathrm{init}},\delta_{\mathrm{init}}]^{p_{\LC}}}
    |\partial_{\theta_i}f^0_{\theta^{\LC}}|^2\,
    \frac{\mathrm{d}^{p_{\LC}}\theta^{\LC}}{(2\delta_{\mathrm{init}})^{p_{\LC}}},
    \label{eq:factorisation}
\end{equation}
where $p_{\LC}$, the number of light-cone parameters, depends on the local circuit geometry but not on the total number of qubits at fixed $L$ and $r$. By Assumption~\ref{ass:nondegen} the integrand is a continuous function of $\theta^{\LC}$ that is not identically zero, so there is an open set $U_{\LC}$ on which $|\partial_{\theta_i}f^0_{\theta^{\LC}}|^2\geq\mu(\beta^\dagger)>0$. The in-cone signal $\mu(\beta^\dagger)$ is proportional to the active-sector light-cone weight $\Tr[P^{\LC}_{q\geq1}\rho_{\LC}]$, so Assumption~\ref{ass:weight} keeps it from vanishing as $1/n$ when the excitation delocalises. Combining the factorisation with the non-degeneracy and sector-weight floor gives
\begin{equation}
    \Mtwo(\theta_i;0)
    \geq
    \frac{\mathrm{vol}(U_{\LC})}{(2\delta_{\mathrm{init}})^{p_{\LC}}}\,
    w_0\,\widetilde{\mu}(\beta^\dagger)
    \eqdef
    c_{\LC}(L,r,\delta_{\mathrm{init}},\beta^\dagger,O_B,w_0)>0,
    \label{eq:lowerbound}
\end{equation}
a constant independent of $n$ at fixed $L$ and $r$, where $\widetilde{\mu}$ is the in-cone signal normalised by the sector weight. The bound is conditional on the fixed non-degenerate background $\beta^\dagger$ and on Assumption~\ref{ass:weight}, and is not claimed to be uniform over degenerate backgrounds or maximally delocalised inputs.
\end{proof}

\subsection{The readout-visible aligned coherence rate}\label{sec:rate}

To make the coherence dependence quantitative, we need an operator-level rate. Let $\Phi^\gamma=I+\gamma\Lop+O(\gamma^2)$ be the small-noise expansion of the single-qubit channel and let $\Pioff$ project onto operators on the off-diagonal block of sector $r$ in the eigenbasis of $\hat N$. The worst-case coherence-contraction rate is,
\begin{equation}
    \lcoh
    \eqdef
    \sup_{\substack{A=\Pioff A\\ \|A\|_F=1}}
    -\mathrm{Re}\,\langle A,\Pioff\Lop(A)\rangle_F ,
    \label{eq:lcohdef}
\end{equation}
the largest decay rate of the restricted generator on the intra-sector off-diagonal subspace in the Frobenius inner product. This worst-case rate is convenient but not fundamental, because it maximises over all off-diagonal directions while the gradient occupies only one. The active gradient of parameter $\theta_i$ is generated by the commutator direction $G_i(\theta)=\Pioff[Z_{q_i},O_B(\ell_i;\theta)]$, the readout-visible mode through which the projected readout actually responds. The fundamental object is the readout-visible aligned rate, the ensemble-weighted decay rate of the restricted generator along this direction,
\begin{equation}
    \lvis(\theta_i)
    \eqdef
    \mathbb{E}_{\theta}\!\left[
    \frac{-\mathrm{Re}\,\langle G_i(\theta),\Pioff\Lop(G_i(\theta))\rangle_F}
         {\langle G_i(\theta),G_i(\theta)\rangle_F}
    \right].
    \label{eq:lvisdef}
\end{equation}
By construction, the aligned rate is bounded above by the worst-case rate, since for every $\theta$ the supremum in Equation~\eqref{eq:lcohdef} runs over a set that contains the normalised gradient direction, so
\begin{equation}
    \lvis(\theta_i)\leq\lcoh
    \label{eq:visbound}
\end{equation}
for every active parameter. The two coincide when the restricted generator acts isotropically on the off-diagonal block, because every normalised direction then decays at the common rate. The distinction is invisible for an isotropic family but decisive for structured noise, where a channel can contract directions orthogonal to $G_i$ and so carry a large $\lcoh$ while leaving $\lvis$ near zero.

Applying the noise once per layer accumulates the aligned contraction into the layered coherence defect $\Dcohlay(\gamma)=\gamma L\,\lvis+O((\gamma L)^2)$, which reduces to $\gamma L\,\lcoh$ in the restricted-isotropic case. For the four restricted-isotropic channels used in the main sweep, the worst-case and aligned rates agree to four significant figures, with $\lcoh=1.000$ for amplitude damping, $3.996$ for dephasing, $6.648$ for depolarising and $7.973$ for the $U(1)$-breaking $X$-error control. Here, restricted-isotropic refers to the action on the intra-sector off-diagonal block rather than to ordinary channel isotropy.

\begin{theorem}[Perturbative coherence-loss law]\label{thm:perturb}
For an active parameter $\theta_i$ and a phase-covariant Markovian noise channel of rate $\gamma$ with Frobenius-dissipative restricted generator acting once per layer \cite{holevo1996,lindblad1976,gks1976}, and under the first-order alignment condition that the noisy derivative acquires no leading-order component along readout-visible directions orthogonal to $G_i$, the leading-order degradation is a sum of layer-resolved aligned contractions,
\begin{equation}
    \Mtwo(\theta_i;\gamma)
    =
    \Mtwo(\theta_i;0)
    \Bigl[1-2\,\gamma\textstyle\sum_{\ell=1}^{L}\lvis^{(\ell)}(\theta_i)+O\bigl((\gamma L\,\lcoh)^2\bigr)\Bigr].
    \label{eq:perturb}
\end{equation}
In the restricted-isotropic family, where every off-diagonal direction decays at the common rate, the layer sum collapses to a scalar accumulation,
\begin{equation}
    \Mtwo(\theta_i;\gamma)
    =
    \Mtwo(\theta_i;0)
    \bigl[1-2\,\gamma L\,\lvis(\theta_i)+O\bigl((\gamma L\,\lcoh)^2\bigr)\bigr],
    \label{eq:perturbscalar}
\end{equation}
with the bounded alignment factor $\omega_i=\lvis(\theta_i)/\lcoh\in[0,1]$ equal to one for isotropic generators and strictly less than one when the gradient mode is misaligned with the fastest-decaying direction.
\end{theorem}

Because three closely related alignment quantities appear in the paper, Table~\ref{tab:alignment} collects their definitions. The theoretical ratio $\omega_i=\lvis/\lcoh$ is the object the theory identifies as fundamental, its perturbative empirical estimate $\weff$ is read off the small-noise experiments, and the finite-noise multiplier $\aeff$ is an a posteriori diagnostic measured in the presence of finite noise.

\begin{table}[t]
\caption{The three alignment quantities used in the paper. The first is a theoretical ratio, the second its perturbative empirical estimate, and the third a finite-noise empirical multiplier.}\label{tab:alignment}
\centering
\begin{tabular}{@{}lll@{}}
\toprule
Symbol & Definition & Role\\
\midrule
$\omega_i=\lvis/\lcoh$ & ratio of aligned to worst-case rate & theoretical alignment factor, in $[0,1]$\\
$\weff$ & $\widetilde{\Delta}/(2\gamma L\,\lcoh)$ in the perturbative window & empirical estimate of $\omega_i$ at small noise\\
$\aeff$ & finite-noise weight fitted to paired degradation & a posteriori diagnostic of residual alignment\\
\bottomrule
\end{tabular}
\end{table}

\begin{proof}[Proof of Theorem~\ref{thm:perturb}]
Let $\Phi^\gamma=I+\gamma\Lop+O(\gamma^2)$ be the small-noise expansion of the layer channel and $\Pioff$ the projector onto the off-diagonal block of sector $r$. The worst-case rate is Equation~\eqref{eq:lcohdef} and the aligned rate is the Rayleigh quotient of Equation~\eqref{eq:lvisdef} with respect to the gradient mode $G_i$. Inserting the expansion into the commutator form of the derivative and projecting onto $G_i(\theta)=\Pioff[Z_{q_i},O_B(\ell_i;\theta)]$ gives, at each layer at which the noise acts, a first-order contraction of the gradient mode propagated to that layer at the pointwise rate $\lvis^{(\ell)}(\theta)$. The single-layer contraction along the mode is the Rayleigh quotient by definition, and for a Frobenius-dissipative restricted generator, each such quotient is a non-negative spectral decay rate. Because the noise acts once per layer and the first-order terms add, the leading-order degradation is the layer sum $2\gamma\sum_{\ell=1}^{L}\lvis^{(\ell)}(\theta_i)$ of Equation~\eqref{eq:perturb} after the $\theta$ ensemble average.

Writing the first-order noisy derivative as $\partial_{\theta_i}f^\gamma_\theta=\partial_{\theta_i}f^0_\theta(1-\gamma\sum_\ell\lvis^{(\ell)})+\gamma\,\delta_i$, where $\delta_i$ collects the first-order response along readout-visible directions orthogonal to $G_i$, the square produces a diagonal term and cross terms linear in $\delta_i$. The first-order alignment condition is exactly the statement that the ensemble average of these cross terms vanishes at leading order, which holds identically in the restricted-isotropic family where $\delta_i=0$ because every off-diagonal direction decays at the common rate. The collapse to the scalar accumulation $\gamma L\,\lvis$ of Equation~\eqref{eq:perturbscalar} requires the aligned rate to be layer-independent, which holds in that family because the interleaved unitaries rotate the gradient mode without changing its decay rate. Outside it, the interleaved unitaries can carry the mode through directions of differing aligned rate, the layer-resolved rates no longer coincide, and the scalar $\lcoh$ can no longer stand in for the layer-averaged aligned contraction. The alignment ratio $\omega_i=\lvis/\lcoh$ inherits the bounds $[0,1]$ from Equation~\eqref{eq:visbound}.
\end{proof}

\subsection{Noise channels and structural diagnostics}

The restricted-isotropic family comprises four single-qubit channels applied independently to every qubit once per layer, each parameterised so the per-layer rate $\gamma$ enters the Kraus probabilities linearly at small noise. The term isotropic refers to the induced action on the intra-sector off-diagonal block in the $r=1$ sector, on which the restricted generator contracts every normalised direction at the same rate so that $\lvis=\lcoh$. Amplitude damping models energy relaxation, dephasing applies a $Z$ flip and models transverse decoherence, depolarising replaces the qubit state by the maximally mixed state, and the $X$-error channel applies a bit flip and serves as the $U(1)$-breaking control. The structured family comprises inhomogeneous dephasing, site-dependent amplitude damping, biased Pauli noise, a coherent-dissipative mixed channel, and correlated two-site dephasing, the last built so its decaying coherence direction is orthogonal to the gradient mode, making it the zero-alignment control. The worst-case rates used in the regressions are $\lcoh=1.000$ for amplitude damping, $3.996$ for dephasing, $6.648$ for depolarising, $7.973$ for the $X$-error control, $3.996$ for inhomogeneous dephasing, $1.000$ for site-dependent amplitude damping, $6.000$ for biased Pauli noise, $1.500$ for the coherent-dissipative mix, and $4.500$ for correlated dephasing. The aligned rate is evaluated directly from the generator only for the correlated-dephasing control, where the calculation confirms $\lvis\approx0$ against $\lcoh=4.5$. The comparison metrics, average gate infidelity $1-F_{\mathrm{avg}}$, unitarity loss $1-u$, the diamond-distance upper bound, purity loss, and state off-diagonal loss, each enter the predictor comparison paired with $\log(\gamma L)$ exactly as the coherence rate does.

\subsection{Numerical protocol and regression analysis}

All simulations evolve full density matrices, with charge-sector preservation verified to floating-point precision before noise is applied. Noise acts once per layer as a Markovian Kraus map \cite{lindblad1976,gks1976,breuerpetruccione2002,nielsenchuang2010}. The small-box prior draws every trainable angle independently and uniformly from $[-\delta_{\mathrm{init}},\delta_{\mathrm{init}}]$ with $\delta_{\mathrm{init}}=0.5$, while the hopping parameters remain fixed. Derivatives use the parameter-shift rule with shifts of $\pi/2$. For the small-noise experiments the noiseless and noisy evaluations share common random numbers, which removes the dominant sampling variance from the ratio defining $\DM$ and is essential in the regime $\gamma L\,\lcoh\ll1$. The main sweep uses the noise grid $\gamma L\in\{0.01,0.03,0.05,0.1,0.2,0.3\}$ per channel-depth pair, the perturbative experiment uses $\gamma L\in\{0.001,0.003,0.005,0.01,0.02,0.03,0.05\}$ with window cutoffs $\gamma L\,\lcoh\leq\{0.01,0.03,0.1,0.3,1.0\}$, and the size study spans $n\in\{6,8,10\}$ at $L=3$, the largest size tractable for full-density-matrix evolution. Activity classification uses a noiseless preflight with a threshold of $10^{-10}$ on $\Mtwo$, far above the inactive floor near $10^{-30}$, which is the square of a double-precision gradient amplitude near $10^{-15}$. Bootstrap confidence intervals use $1000$ resamples.

The statistical analysis fits ordinary least squares to $\log\DM$, with the main model regressing on $\log(\gamma L)$ and $\log\lcoh$ and the reduced model on $\log(\gamma L)$ alone. Model comparison uses the Akaike information criterion, reported as differences relative to the best model. Robust uncertainty is assessed through HC0 standard errors, and a cluster bootstrap over channel-depth cells with $1000$ replicates, a mixed-effects model with random intercepts for channel and depth checks that the exponents are not group-structure artefacts, and a permutation test reassigns the channel rates to labels across $1000$ permutations. Cross-validation uses leave-one-channel-out, leave-one-depth-out and leave-one-noise-level-out folds.

\section{Results}\label{sec:results}

\subsection{Active gradients localise to the light cone}\label{sec:res-lightcone}

Figure~\ref{fig:activity} confirms the prediction numerically. The noiseless activity map $\log_{10}\Mtwo(\theta_{\ell,j};0)$ at $n=8$ and $r=1$ shows that active parameters sit inside the backward light cone of the readout, the active set grows with depth, and every parameter outside the cone rests at the numerical floor near $10^{-30}$. The red outline is the analytic light cone inferred from the brickwork causal structure, and it coincides cell for cell with the numerically active region. Across the maps the median active-to-inactive ratio of $\Mtwo$ exceeds $10^{27}$ at both depths, a structural separation rather than a faint signal lifted above a sampling floor. Layer $\ell=0$ stays inactive because the frozen initial state and the small-box prior make the first rotation layer act trivially on the gradient at leading order.

\begin{figure}[t]
\centering
\includegraphics[width=0.95\textwidth]{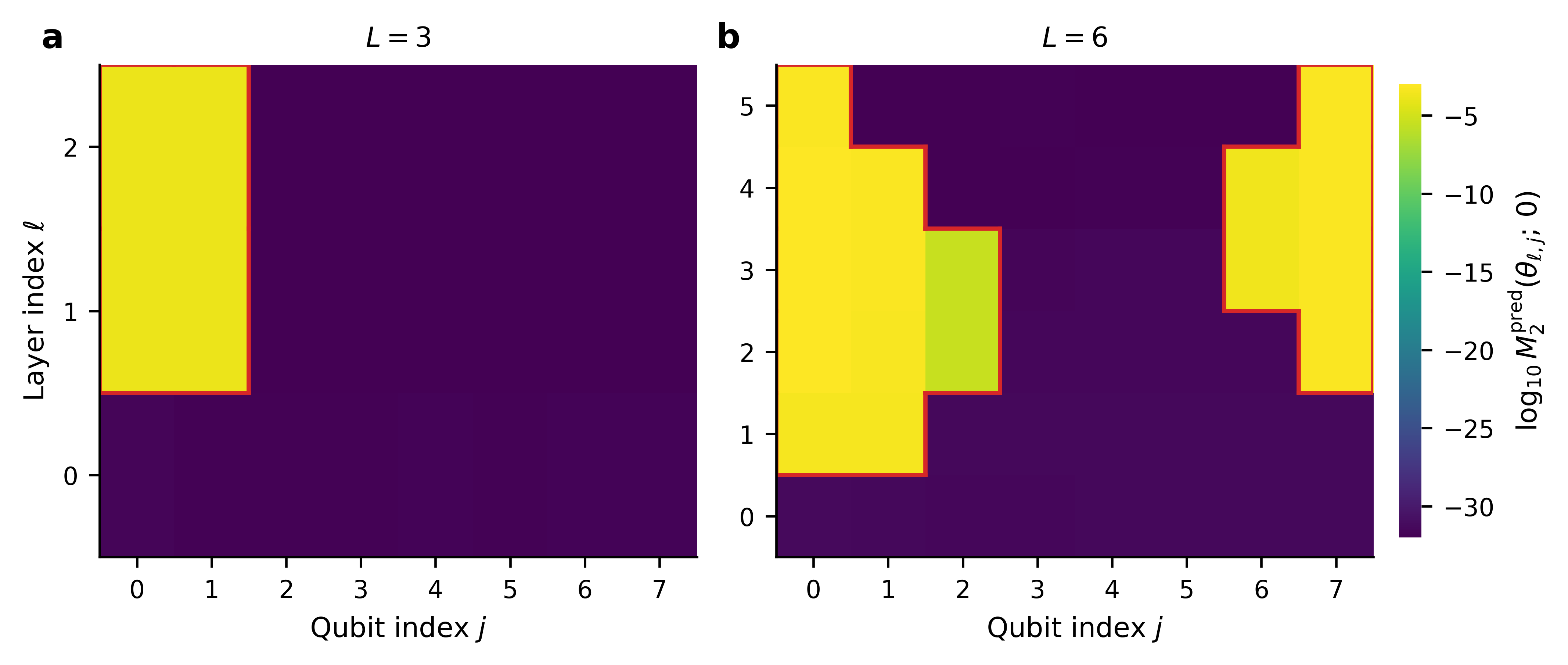}
\caption{Backward-light-cone localisation of active gradients. Numerically computed prediction-gradient second moment $\log_{10}\Mtwo(\theta_{\ell,j};0)$ for $n=8$ qubits in charge sector $r=1$ at depths $L=3$ (a) and $L=6$ (b). The red outline marks the analytic readout backward light cone inferred from the brickwork causal structure, a theoretical prediction rather than a contour fitted to the heatmap. Active parameters lie inside the cone while inactive parameters rest at the numerical floor near $10^{-30}$, giving a median active-to-inactive ratio above $10^{27}$. The $L=6$ cone wraps around the cycle because of the periodic boundary.}\label{fig:activity}
\end{figure}

\subsection{The perturbative law holds at small noise}\label{sec:res-perturb}

Figure~\ref{fig:perturb} validates the perturbative prediction using paired small-noise simulations at $n=8$, $L=4$ and $r=1$, in which the noiseless and noisy second moments share common random numbers so that the ratio defining $\DM$ is far less sensitive to sampling-floor effects. The paired degradation $\widetilde{\Delta}=1-\Mtwo(\theta_i;\gamma)/\Mtwo(\theta_i;0)$ is linear in $\gamma L$ for all four channels, with per-channel $R^2$ above $0.986$. The inferred alignment estimate $\weff$ stays within $[0,1]$ and approaches channel-specific limits, near one for amplitude damping and near one-half for dephasing. Pooled fits over increasingly tight perturbative windows recover exponents close to the target value of 1.

\begin{figure}[t]
\centering
\includegraphics[width=\textwidth]{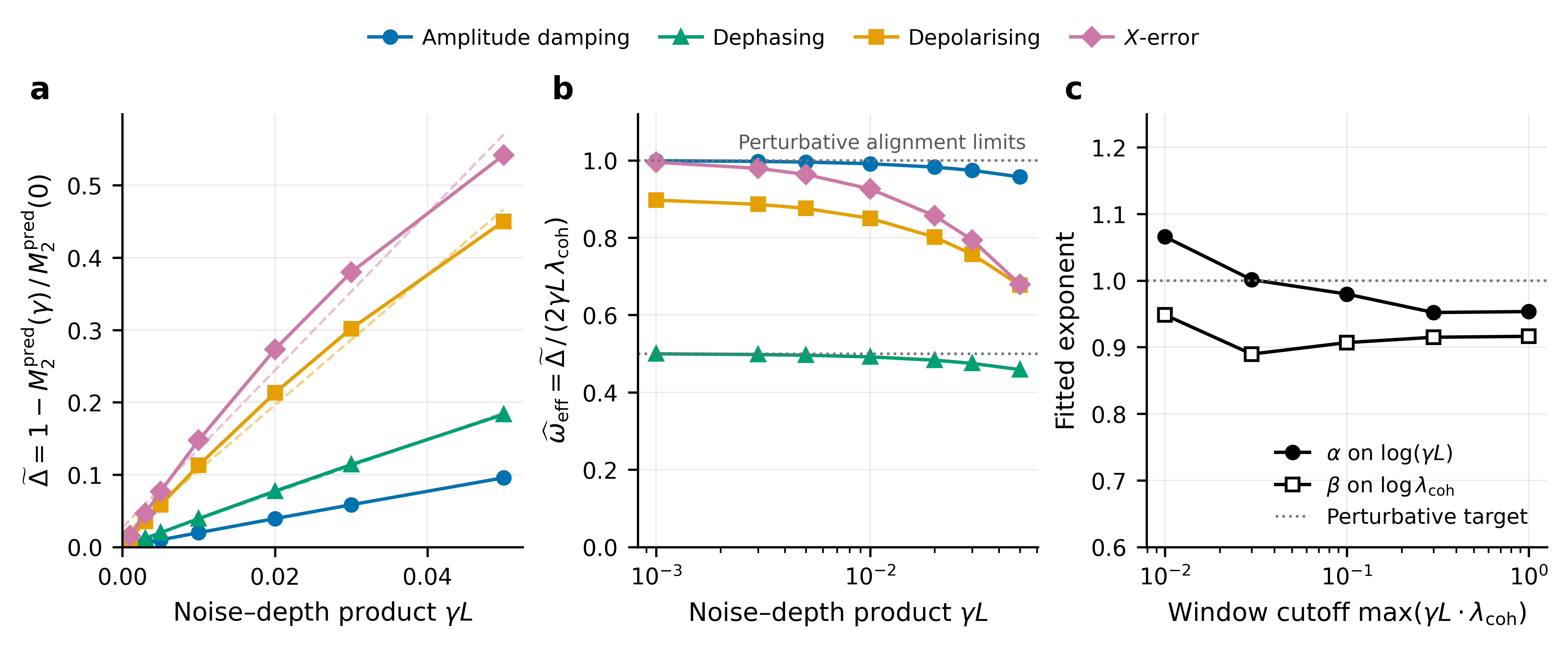}
\caption{Perturbative sector-coherence loss. Paired small-noise simulations at $n=8$, $L=4$, $r=1$ for the four restricted-isotropic channels. (a) Per-channel linearity of $\widetilde{\Delta}=1-\Mtwo(\gamma)/\Mtwo(0)$ in the noise-depth product $\gamma L$, with dashed small-noise tangents and per-channel $R^2$ above $0.986$. (b) The perturbative alignment estimate $\weff$ approaches channel-specific limits at $1$ and $1/2$ and remains in $[0,1]$. (c) Pooled exponents under tightening window cutoffs on $\gamma L\,\lcoh$, with the exponent on $\log(\gamma L)$ contracting towards one and the coherence exponent between $0.89$ and $0.95$.}\label{fig:perturb}
\end{figure}

\subsection{Finite-noise degradation follows accumulated coherence loss}\label{sec:finitenoise}

The perturbative law predicts the leading-order behaviour, but the operationally relevant regime extends to finite noise where higher-order terms are no longer negligible. The main sweep covers $n=8$, $r=1$, depths $L\in\{3,4,5,6\}$, the four restricted-isotropic channels and six noise levels per channel-depth pair, giving $96$ settings. Each setting is itself an ensemble estimate of $\DM$ with bootstrap confidence intervals, so the regression sample size of $96$ counts independent physical settings. The central regression is the log-linear model,
\begin{equation}
    \log\DM
    =
    c+\alpha\log(\gamma L)+\beta\log\lcoh+\epsilon ,
    \label{eq:fit}
\end{equation}
and on the pooled data it gives $\alpha=1.002$, $\beta=0.964$ and $R^2=0.9789$, so that,
\begin{equation}
    \DM\propto(\gamma L)^{1.002}\,\lcoh^{0.964}.
    \label{eq:law}
\end{equation}
Both exponents sit close to one, which is the accumulated form of the perturbative law specialised to the restricted-isotropic family, where $\lvis=\lcoh$ for every channel so that regressing on $\lcoh$ is the same as regressing on the aligned rate. Because Theorem~\ref{thm:perturb} is stated for phase-covariant noise, the primary theorem-supporting fit uses only the three symmetry-preserving channels, giving $\alpha=1.006$, $\beta=0.889$ and $R^2=0.978$, statistically indistinguishable from the pooled result. When the symmetry-breaking $X$-error channel is held out and predicted from the three-channel fit, the held-out coefficient of determination is $0.949$, comparable to the within-family folds, which shows that it falls on the same scalar law without being part of the theorem-supporting regression.

The finite-noise degradation is therefore organised, to good accuracy, by the single accumulated variable $\gamma L\,\lcoh$, and Figure~\ref{fig:collapse} displays the collapse. All four depths fall on a common line, so the relevant variable is the accumulated product $\gamma L$ rather than depth itself. Adding an explicit $\log L$ term changes nothing of substance, with a fitted coefficient of $-0.031$, a $p$-value of $0.71$ and no improvement in $R^2$.

\begin{figure}[t]
\centering
\includegraphics[width=\textwidth]{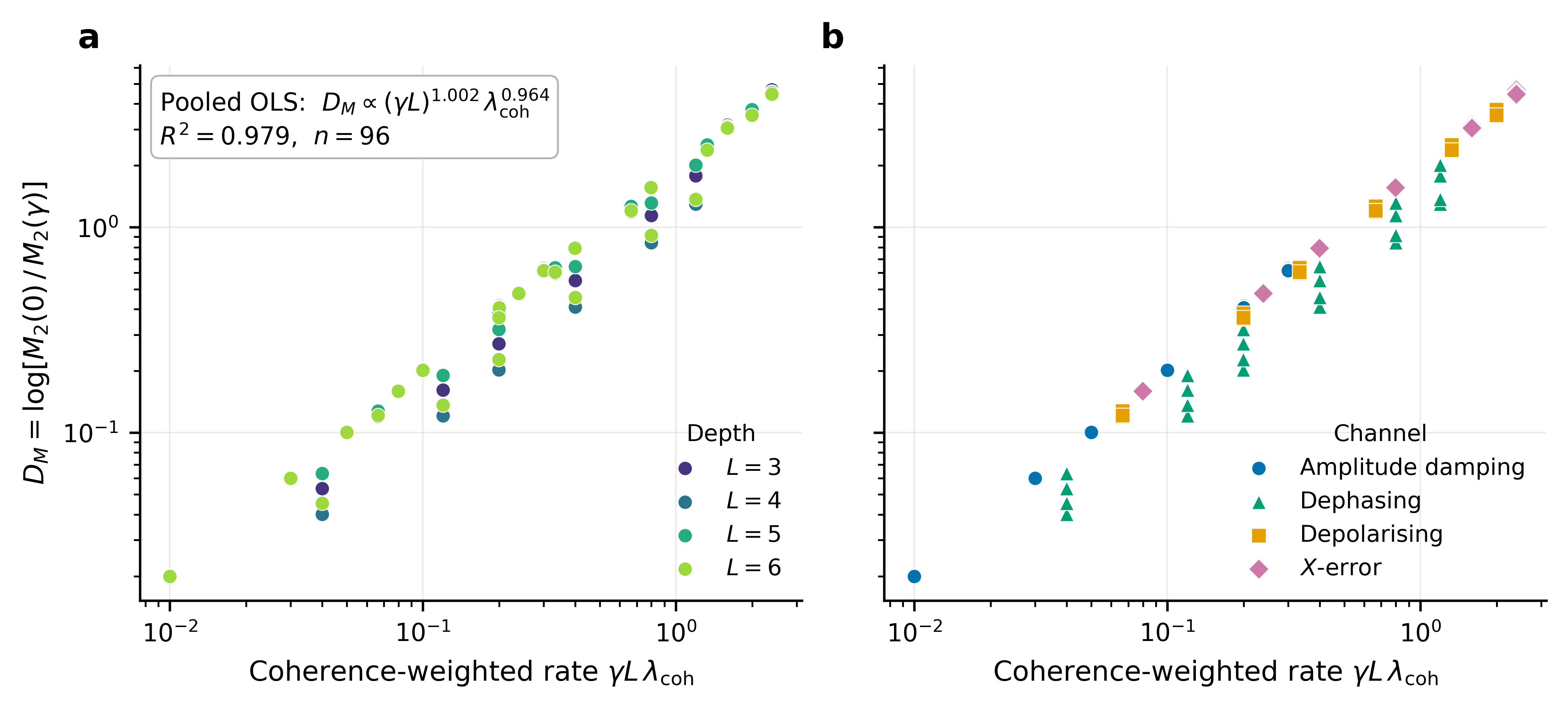}
\caption{Finite-noise accumulated coherence collapse. Gradient degradation $\DM$ against the coherence-weighted noise rate $\gamma L\,\lcoh$ for the main sweep at $n=8$, $r=1$, $L\in\{3,4,5,6\}$, four restricted-isotropic channels and six noise levels per channel-depth pair, giving $n=96$ settings. (a) Points coloured by circuit depth. (b) The identical points coloured by the noise channel. The pooled fit gives $\DM\propto(\gamma L)^{1.002}\lcoh^{0.964}$ with $R^2=0.979$. The dephasing offset band in (b) reflects residual channel structure absorbed by the alignment-weighted refinement.}\label{fig:collapse}
\end{figure}

A fair reading of the headline fit asks how much of the explained variance is already captured by noise depth alone. The noise-depth-only regression on $\log(\gamma L)$ explains $R^2=0.670$. Adding the coherence rate raises $ R^2$ to $ 0.979$ and reduces the root-mean-square error from $0.810$ to $0.205$. The unique contribution of the coherence rate is captured by the partial coefficient of determination,
\begin{equation}
    R^2_{\mathrm{partial}}
    =
    \frac{R^2_{\mathrm{full}}-R^2_{\gamma L\text{-only}}}{1-R^2_{\gamma L\text{-only}}}
    =0.936,
    \label{eq:partial}
\end{equation}
so after accounting for noise depth the sector-coherence rate explains roughly $94\%$ of the variance that noise depth alone leaves unresolved. The fit is stable under heteroscedasticity-robust standard errors, a cluster bootstrap over channel-depth cells, mixed-effects modelling with random intercepts for channel and depth, and leave-one-depth-out and leave-one-noise-level-out validation, with held-out $R^2$ between $0.93$ and $0.99$.

\subsection{Sector coherence outperforms standard diagnostics}\label{sec:predictor}

The coherence rate would be of limited interest if any reasonable channel-strength metric organised the same data equally well. Figure~\ref{fig:predictor} compares the main model against the standard alternatives, each pairing $\log(\gamma L)$ with one channel metric on the identical pooled response. Pairing $\gamma L$ with $\lcoh$ gives $R^2=0.979$, against $0.961$ for purity loss, $0.949$ for state off-diagonal loss, roughly $0.808$ for average infidelity, unitarity loss and the diamond-distance proxy, and $0.670$ for noise depth alone. The same ordering appears in the Akaike information criterion, where the coherence-rate model improves on the gate-level metrics by more than $210$ units. The advantage is not that the model carries an extra predictor. Average infidelity, unitarity, purity and the generic off-diagonal loss of the full state all average over degrees of freedom that play no part in the active-gradient response, whereas the coherence rate is the contraction rate on the intra-sector off-diagonal subspace through which the projected readout actually responds. A diagnostic matched to the response-carrying subspace outperforms diagnostics that dilute the same physics across the whole space, and the result survives a permutation null with one-sided $p\simeq0.045$.

\begin{figure}[t]
\centering
\includegraphics[width=\textwidth]{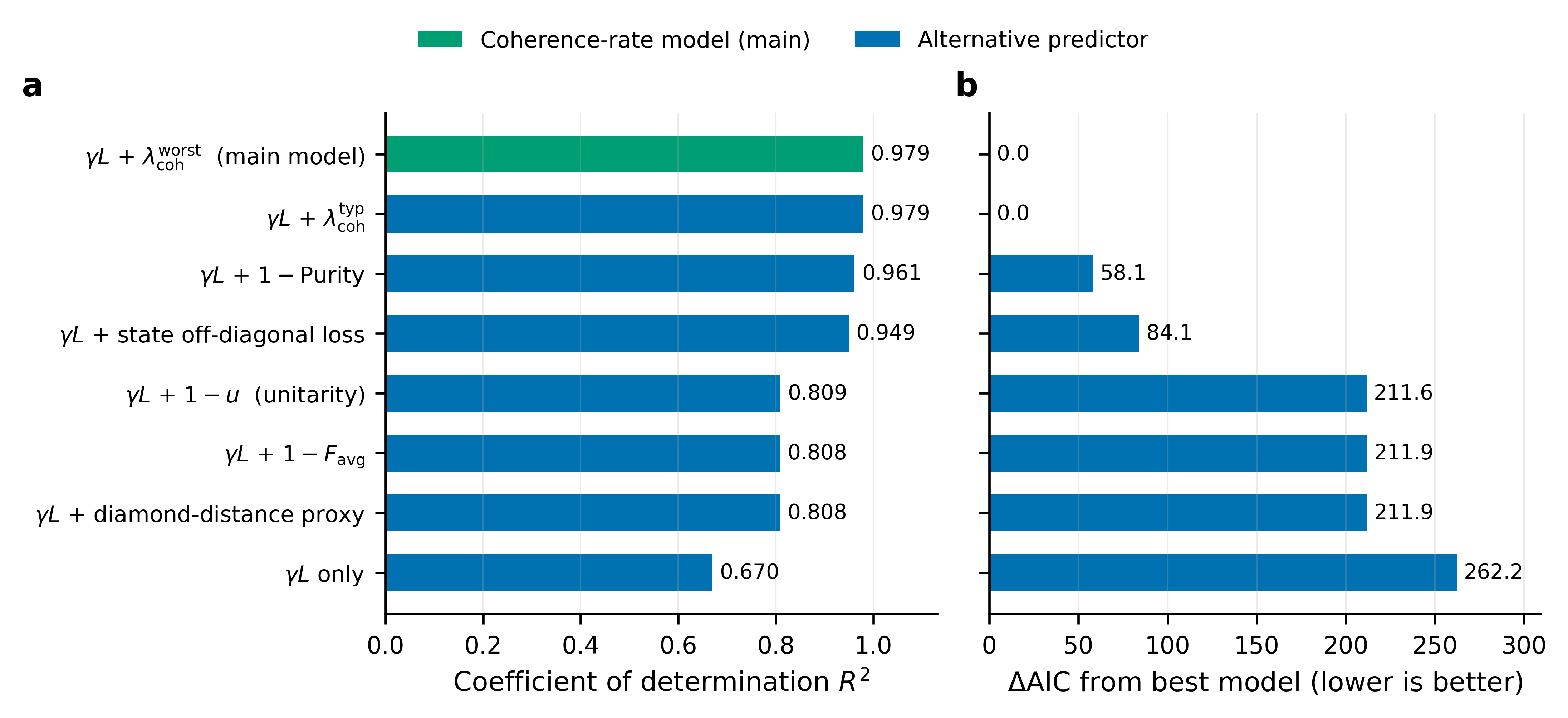}
\caption{Sector coherence outperforms standard noise diagnostics. Predictor comparison for the $\log\DM$ regression on the pooled dataset with $n=96$ settings. (a) Coefficient of determination $R^2$ by predictor set, each alternative pairing $\log(\gamma L)$ with one channel metric, with the coherence-rate model highlighted. (b) The difference in the Akaike information criterion from the best model, with lower being better. The coherence-rate model improves on the gate-level metrics by more than $210$ AIC units and on the noise-depth baseline by more than $260$.}\label{fig:predictor}
\end{figure}

\subsection{Structured noise reveals the aligned rate}\label{sec:structured}

The restricted-isotropic family is the reference case, because its restricted generators contract every off-diagonal direction at the same rate. Structured noise breaks this degeneracy and tests whether the bare scalar $\lcoh$ is still the right object. The structured study uses five anisotropic channels at $n=8$, $L=3$ and $r=1$. These are inhomogeneous dephasing, site-dependent amplitude damping, biased Pauli noise, a coherent-dissipative mixed channel, and correlated two-site dephasing. Figure~\ref{fig:aniso} presents the analysis.

Four of the five structured channels follow the same power law in $\gamma L$ as the restricted-isotropic family, each along its own coherence-weighted line. The fifth, correlated dephasing, is the decisive control. It is built so that its decaying coherence direction runs orthogonal to the gradient-carrying mode $G_i(\theta)$ across the ensemble, which sends the aligned rate to $\lvis\approx0$ even as the worst-case rate stays large at $\lcoh=4.5$. The aligned theory predicts negligible degradation in advance, and the simulation bears this out, with the measured degradation pinned at the numerical floor, $|\DM|<3\times10^{-14}$ at every noise level. This is a direct test of the inequality $\lvis\leq\lcoh$ at its extreme, where a channel saturates a large worst-case rate while leaving the readout-visible rate near zero. As a zero-alignment control, it shows that the bare worst-case rate cannot, on its own, be the fundamental variable.

Pooling the isotropic and anisotropic data with the control left out of the fit yields $\alpha=1.001$, $\beta=1.040$ and $R^2=0.964$ across $126$ settings, with noticeably more residual scatter than the isotropic collapse. The drift of the fitted exponent from the isotropic $0.964$ towards $1.040$, together with the wider scatter, is the signature one expects when the regression is written in $\lcoh$ but the response is governed by $\lvis$. The structured channels carry $\lvis<\lcoh$ by varying amounts that the scalar regressor cannot resolve. The finite-noise empirical multiplier $\aeff$, an a posteriori estimate of $\omega_i=\lvis/\lcoh$ measured directly from the paired degradation data and bounded by one across all settings, absorbs that gap. Weighting the rate by it raises the main pooled regression from $R^2=0.979$ to $0.991$ and cuts the root-mean-square error of $\log\DM$ from $0.205$ to $0.131$. Because $\aeff$ is read off the response, it is a post hoc diagnostic rather than an independent predictor, so the decisive evidence for the aligned rate remains the zero-alignment control. In matched settings, the scalar coherence model gives weaker organisation outside the single-excitation sector, with $R^2$ of $0.71$ at $r=2$ and $0.89$ at $r=3$, so the scalar law is treated as a single-excitation, restricted-isotropic result, with direct evaluation of the aligned rate the appropriate route for higher sectors.

\begin{figure}[t]
\centering
\includegraphics[width=0.95\textwidth]{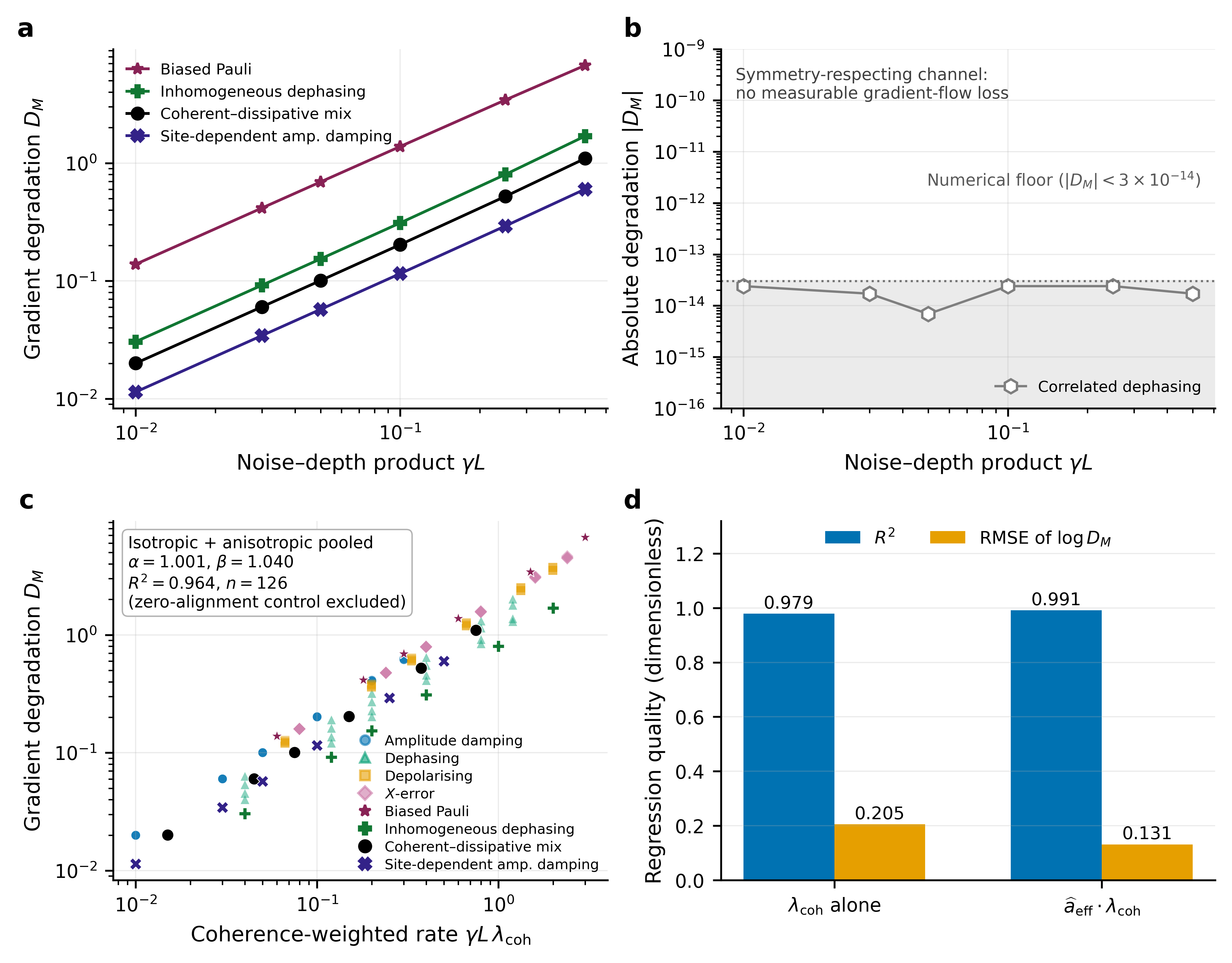}
\caption{Structured noise and readout-visible aligned coherence. Anisotropic-noise study at $n=8$, $L=3$, $r=1$, pooled with the isotropic dataset. (a) Four structured channels follow the same power law in $\gamma L$, each along its own coherence-weighted line. (b) Correlated dephasing shown separately as a zero-alignment control. Despite a non-zero worst-case rate $\lcoh=4.5$, its aligned rate satisfies $\lvis\approx0$, so the measured degradation satisfies $|\DM|<3\times10^{-14}$ at every noise level. (c) Pooled isotropic and anisotropic collapse with the control excluded, giving $\alpha=1.001$, $\beta=1.040$ and $R^2=0.964$ over $n=126$ settings. (d) Weighting the rate by the finite-noise multiplier $\aeff$ improves the regression from $R^2=0.979$ to $0.991$ and reduces the root-mean-square error of $\log\DM$ from $0.205$ to $0.131$.}\label{fig:aniso}
\end{figure}

\subsection{Size and topology checks}\label{sec:size}

Theorem~\ref{thm:locality} predicts that at fixed depth and charge sector, the degradation law is insensitive to the boundary conditions of the graph and to the total number of qubits. Both predictions hold. Replacing the cycle with an open chain at $n=8$ leaves the per-channel degradation unchanged, and a topology-pooled regression across $45$ settings yields $\alpha=1.016$, $\beta=0.943$ and $R^2=0.991$. A size check over $n\in\{6,8,10\}$ at $L=3$ collapses onto the common line with $\alpha=1.007$, $\beta=0.909$ and $R^2=0.983$, and an explicit $\log n$ term yields a negligible size exponent $\eta\simeq0.043$. The size sweep also bears on Assumption~\ref{ass:weight}, since the noiseless active second moment shows no $1/n$ suppression across the three sizes, remaining of order $10^{-5}$ to $10^{-4}$, consistent with the fixed input keeping finite light-cone sector weight as $n$ grows. The degradation is exactly $n$-invariant for amplitude damping and dephasing, while the only size trend in depolarising noise is explained by its worst-case rate rising with $n$, evaluated per size, from $\lcoh=5.32$ at $n=6$ to $7.97$ at $n=10$.

\section{Discussion}\label{sec:discussion}

Equivariance determines where gradients can live, and readout-visible sector coherence determines whether those gradients survive noise. The light-cone reduction pins the support of the active response to the sector-restricted backward light cone of the projected readout and keeps its noiseless second moment away from zero independently of system size. The perturbative law then names the off-diagonal mode $G_i$ that carries the amplitude of that response, together with the aligned rate $\lvis$ that controls its contraction. The finite-noise simulations show that this mechanism stays predictive well beyond the first-order window, with near-linear exponents in both accumulated noise depth and sector-coherence contraction. Nothing in the finite-noise regime was guaranteed by the first-order expansion, so the persistence of the near-linear law is a genuine empirical finding about this model family rather than a corollary of the theory.

The same reasoning explains why standard channel diagnostics underperform. Average infidelity, unitarity loss, purity loss, and diamond-distance proxies each measure a legitimate notion of channel disturbance, but they average over degrees of freedom that need not participate in the active-gradient response. The sector-coherence rate is tied to the intra-sector off-diagonal block through which the projected readout responds, so its predictive advantage comes from a physical match between the diagnostic and the response subspace rather than from a statistical artefact of fitting. The point generalises beyond this model. A trainability diagnostic should be matched to the operator subspace of the response it is meant to predict.

Structured noise shows where the scalar diagnostic must be refined. Correlated dephasing has a nonzero worst-case contraction yet produces no measurable gradient loss because its decaying coherence direction is orthogonal to the gradient mode, and its aligned rate is therefore near zero. This zero-alignment control is the decisive evidence that the aligned rate, not the bare worst-case rate, is the operator-level object that controls degradation, with the inequality $\lvis\leq\lcoh$ becoming an equality only for restricted-isotropic channels. The definition of $\lvis$ as a Rayleigh quotient supplies the route to an independent operator-level metric, and confirming that the directly computed $\lvis$ reproduces the empirical alignment across the structured family is the clearest next step, since the naive generator-spectrum estimator misassigns alignment for correlated dephasing and a correct estimator needs the readout projection built in.

Two practical consequences follow for noisy equivariant quantum models. One is a matter of benchmarking. Reporting symmetry covariance or sector leakage is not enough, since a channel can preserve both while erasing the coherence that carries the gradient, so a useful benchmark should also report how much sector coherence survives within the readout light cone. The other is a matter of design. Given a characterised noise model, hardware-aware ansatz design can weigh the readout light cone, the active sector, and the noise directions that contract the corresponding coherence modes, and so place the gradient-carrying coherence in the most slowly contracting directions on offer.

A few limitations bound what has been shown. The study is scoped to the single-excitation $U(1)$ sector and the restricted-isotropic family, which is the regime where the aligned rate reduces to the worst-case scalar. Higher sectors sit at the edge of that regime and show weaker single-scalar organisation, consistent with their richer off-diagonal structure needing the aligned diagnostic in place of the worst-case scalar. Full density-matrix simulation limits the accessible system sizes, with the size check reaching $n=10$. Coherent unitary miscalibration, which rotates rather than contracts off-diagonal modes, calls for a different treatment. Experimental validation of the full coherence-rate law on hardware is left to future work, and no hardware results are reported here.

\section{Conclusion}\label{sec:conclusion}

We have set out a training law for noisy equivariant quantum neural networks that ties the survival of trainable gradients to two physical effects, causality and coherence. Causality confines the active gradient to the sector-restricted backward light cone of the projected readout. A reduction theorem makes this confinement quantitative, with a lower bound independent of the total qubit number, so that adding idle qubits outside the cone cannot dilute the signal. Coherence then sets the rate at which the gradient fades, through the readout-visible aligned contraction of the off-diagonal sector mode the readout can see. A perturbative open-system analysis turns this into a leading-order law, and density-matrix simulations confirm that it holds well into the finite-noise regime as $\DM\propto(\gamma L)^{1.002}\lcoh^{0.964}$ with $R^2\simeq0.979$.

The single most informative result is the negative control. A correlated-dephasing channel with a large worst-case contraction but a near-zero aligned rate produces no measurable gradient loss, exactly as the aligned theory predicts in advance, and against what any whole-channel strength measure would suggest. Together with the systematic advantage of sector coherence over every standard channel diagnostic, this control identifies readout-visible, aligned sector coherence as the operator-level quantity governing noisy trainability, in place of symmetry covariance or sector population. The wider lesson is that a trainability diagnostic should be matched to the operator subspace of the response it predicts. Within the single-excitation, restricted-isotropic regime that matching yields a clean and reproducible law, and the Rayleigh-quotient definition of the aligned rate marks the route towards higher sectors, structured noise, and eventually hardware, where the same light-cone and coherence structure should continue to organise how gradients survive under noise.

\section*{Acknowledgments}
The authors acknowledge the computational resources provided by the Centre for Visual Computing and Intelligent Systems at the University of Bradford.

\section*{Data availability}\label{sec:data}
The complete codebase and the raw and processed data supporting the findings of this study, including all CSV files, are available in the GitHub repository at \url{https://github.com/ugail/Readout-Visible-Coherence-QNN}. The archive contains the complete code along with the processed data required to reproduce all results reported here.

\section*{Funding}
No funding was received for this work.

\subsection*{Competing interests}
The authors declare no competing interests.

\end{document}